\documentclass[a4paper,english,notitlepage,a4paper,english,superscriptaddress,floatfix,longbibliography]{revtex4-1}
\usepackage[T1]{fontenc}
\usepackage[latin9]{inputenc}
\usepackage{color}
\usepackage{babel}
\usepackage{enumitem}
\usepackage{amsmath}
\usepackage{amsthm}
\usepackage{amssymb}
\usepackage{stmaryrd}
\usepackage{graphicx}
\usepackage[unicode=true,
 bookmarks=true,bookmarksnumbered=true,bookmarksopen=false,
 breaklinks=true,pdfborder={0 0 1},backref=false,colorlinks=true]
 {hyperref}
\hypersetup{
 linkcolor=black, citecolor=black, urlcolor=blue, filecolor=blue, pdfpagelayout=OneColumn, pdfnewwindow=true, pdfstartview=XYZ, plainpages=false}

\makeatletter

\pdfpageheight\paperheight
\pdfpagewidth\paperwidth

\providecommand{\tabularnewline}{\\}

\theoremstyle{plain}
\newtheorem{thm}{\protect\theoremname}

\makeatother

\providecommand{\theoremname}{Theorem}

\begin{document}
\title{Loop corrections in spin models through density consistency}
\author{Alfredo Braunstein}
\affiliation{Politecnico di Torino, Corso Duca Degli Abruzzi 24, Torino, Italy }
\affiliation{Italian Institute for Genomic Medicine, Via Nizza 52, Torino, Italy}
\affiliation{Collegio Carlo Alberto, Via Real Collegio 1, Moncalieri, Italy \&
INFN Sezione di Torino, Via P. Giuria 1, I-10125 Torino, Italy}
\author{Giovanni Catania}
\affiliation{Politecnico di Torino, Corso Duca Degli Abruzzi 24, Torino, Italy }
\author{Luca Dall'Asta}
\affiliation{Politecnico di Torino, Corso Duca Degli Abruzzi 24, Torino, Italy }
\affiliation{Collegio Carlo Alberto, Via Real Collegio 1, Moncalieri, Italy \&
INFN Sezione di Torino, Via P. Giuria 1, I-10125 Torino, Italy}
\begin{abstract}
Computing marginal distributions of discrete or semi-discrete Markov
Random Fields (MRF) is a fundamental, generally intractable, problem
with a vast number of applications on virtually all fields of science.
We present a new family of computational schemes to calculate approximately
marginals of discrete MRFs. This method shares some desirable properties
with Belief Propagation, in particular providing exact marginals on
acyclic graphs; but at difference with it, it includes some loop corrections,
i.e. it takes into account correlations coming from all cycles in
the factor graph. It is also similar to Adaptive TAP, but at difference
with it, the consistency is not on the first two moments of the distribution
but rather on the value of its density on a subset of values. Results
on random connectivity and finite dimensional Ising and Edward-Anderson
models show a significant improvement with respect to the Bethe-Peierls
(tree) approximation in all cases, and with respect to Plaquette Cluster
Variational Method approximation in many cases. In particular, for
the critical inverse temperature $\beta_{c}$ of the homogeneous hypercubic
lattice, the expansion of $\left(d\beta_{c}\right)^{-1}$ around $d=\infty$
of the proposed scheme is exact up to the $d^{-4}$ order, whereas
the two latter are exact only up to the $d^{-2}$ order.
\end{abstract}
\maketitle

\paragraph*{Introduction.}

Markov random fields (MRF) are undirected probabilistic graphical
models in which random variables satisfy a conditional independence
property, so that the joint probability measure can be expressed in
a factorized form, each factor involving a possibly different subset
of variables \citep{lauritzen1996graphical}. Computing marginal distributions
of discrete or semi-discrete Markov Random Fields is a fundamental
step in most approximate inference methods and high-dimensional estimation
problems \citep{mackay2003information}, such as the evaluation of
equilibrium observables in statistical mechanics models. The exact
calculation of marginal distributions is however intractable in general,
and it is common to resort to stochastic sampling algorithms, such
as Monte-Carlo Markov Chain, to obtain unbiased estimates of the relevant
quantities. On the other hand, it is also useful to derive approximations
of the true probability distribution, for which marginal quantities
can be deterministically computed. An important family of approximations
is the one of mean-field (MF) schemes. The simplest is naive mean-field
(nMF), that neglects all correlations between random variables. An
improved MF approximation \citep{opper2001naive}, called Thouless-Anderson-Palmer
(TAP) equations, works well for models with weak dependencies but
it is usually unsuitable for MRFs on diluted models. Here, a considerable
improvement is provided by the Bethe-Peierls approximation or Belief
Propagation (BP), that is exact for probabilistic models defined on
graphs without loops \citep{mezard2009information}. It is a matter
of fact that BP has been successfully employed even in on loopy probabilistic
models both in physics and in applications (see e.g. \citet{turbo}),
yet the lack of analytical control on the effect of loops calls for
novel approaches that could systematically improve with respect to
BP. A traditional way to account for the effect of short loops is
by means of Cluster Variational Methods (CVM), that treat exactly
correlations generated between variables within a finite region $\mathcal{R}$
\citep{kikuchi1951theory,yedidia2003understanding,pelizzola2005cluster}.
The main limitation of CVM resides in its algorithmic complexity,
that grows exponentially with the size of the region $\mathcal{R}$.
A completely different path to systematically improve BP is represented
by loop series expansions \citep{chertkov2006loop,gomez2007truncating,parisi2006loop,altieri2017loop}
in which BP is obtained as a saddle-point in a corresponding effective
field theory. Loop corrections to BP equations can be alternatively
introduced in terms of local equations for correlation functions,
as first suggested for pairwise models \citep{montanari_how_2005}
and later extended to arbitrary factor graphs \citep{mooij2007loop,rizzo2007cavity,ohzeki2013belief}.
This method consists in considering deformed local marginal probabilities
on a ``cavity graph'', obtained removing a factor node (i.e. interaction),
and imposing a consistency condition on single-node marginals. On
trees, BP equations are recovered, whereas on loopy graphs the obtained
set of equations is strongly under-determined and requires additional
constraints. Linear-response relations were exploited to this purpose
in \citep{montanari_how_2005}, even though other moment closure methods
are possible \citep{mooij2007loop}. 

A different approach to approximate inference exploits the properties
of multivariate Gaussian distributions, that have the advantage of
retaining information on correlations albeit allowing for explicit
calculations. In particular, Expectation Propagation (EP, which can
be thought as an adaptive variant of TAP) is a very successful algorithmic
technique in which a tractable approximate distribution is obtained
as the outcome of an iterative process in which the parameters of
a multivariate Gaussian are optimized by means of local moment matching
conditions \citep{minka_expectation_2001,opper_adaptive_2001}. EP
has been applied to problems involving discrete random variables by
employing atomic measures \citep{opper2005expectation}. In the present
work, we put forward a new family of computational schemes to calculate
approximate marginals of discrete MRFs. We exploits the flexibility
of multivariate Gaussian approximation methods but, unlike EP and
inspired by beliefs marginalization condition in BP, we impose that
marginals on discrete variables are locally consistent, a condition
that we call \emph{density consistency}. When the underlying graph
is a tree, the set of equations produced is equivalent to BP. As for
Ref.\citep{montanari_how_2005}, the density consistency condition
leaves an under-determined system of equations on loopy graphs, that
can be solved once supplemented with a further set of closure conditions.
As a result of employing Gaussian distributions, higher order correlation
functions between neighbors of a given variable are, at least partially,
taken into account.

\paragraph*{The Model.}

Consider a factorized distribution of binary variables $x_{1},\dots,x_{n}\in X=\left\{ -1,1\right\} $
for arbitrary \emph{positive factors $\psi_{a}:X_{a}\to\mathbb{R}_{+}$,
}each depending on a sub-vector $\boldsymbol{x}_{a}=\left\{ x_{i}\right\} _{i\in\partial a}\in X_{a}$
\begin{equation}
p\left(\boldsymbol{x}\right)=\frac{1}{z}\prod_{a\in A}\psi_{a}\left(\boldsymbol{x}_{a}\right).\label{eq:factorization}
\end{equation}
The bipartite graph $G=\left(V,E\right)$ with $V=I\cup A$ the disjoint
union of variable indices $I=\left\{ 1,\dots,n\right\} $ and factor
indices and $E=\left\{ \left(ia\right):i\in\partial a\right\} $,
is called the \emph{factor graph} of the factorization \eqref{eq:factorization},
and as we will see some of its topological features are crucial to
devise good approximations. Particular important cases of \eqref{eq:factorization}
include e.g. Ising spin models, many neural network models and the
uniform distribution of solutions of $k-$SAT formulas. Computing
marginal distributions from \eqref{eq:factorization} is in general
NP-Hard (i.e. computationally intractable). 

\paragraph*{Density Consistency.}

Following Gaussian Expectation Propagation (otherwise called adaptive
TAP or Expectation Consistency) \citep{opper_adaptive_2001,minka_expectation_2001},
we will approximate an intractable $p\left(\boldsymbol{x}\right)$
by a Normal distribution $g\left(\boldsymbol{x}\right)$. To do so,
we will replace each $\psi_{a}\left(\boldsymbol{x}_{a}\right)$ by
an appropriately defined multivariate normal distribution $\phi_{a}\left(\boldsymbol{x}_{a}\right)=\mathcal{N}\left(\boldsymbol{x}_{a};\boldsymbol{\mu}^{a},\boldsymbol{\Sigma}^{a}\right)$.
Parameters $\boldsymbol{\mu}^{a},\boldsymbol{\Sigma}^{a}$ will be
selected as follows. First define
\[
g\left(\boldsymbol{x}\right)=\frac{1}{z_{g}}\prod_{a}\phi_{a}\left(\boldsymbol{x}_{a}\right)=\frac{1}{z'}e^{-\frac{1}{2}\left(\boldsymbol{x}-\boldsymbol{\mu}\right)^{T}\boldsymbol{\Sigma}^{-1}\left(\boldsymbol{x}-\boldsymbol{\mu}\right)}.
\]
and $g^{\left(a\right)}\left(\boldsymbol{x}\right)=\frac{1}{z_{a}}g\left(\boldsymbol{x}\right)\sum_{\hat{\boldsymbol{x}}_{a}\in X_{a}}\delta\left(\boldsymbol{x}_{a}-\hat{\boldsymbol{x}}_{a}\right)\frac{\psi_{a}\left(\hat{\boldsymbol{x}}_{a}\right)}{\phi_{a}\left(\hat{\boldsymbol{x}}_{a}\right)}$,
$g^{\left(i\right)}\left(\boldsymbol{x}\right)=\frac{1}{z_{i}}g\left(\boldsymbol{x}\right)\sum_{\hat{x}_{i}\in X}\delta\left(x_{i}-\hat{x}_{i}\right)$
as auxiliary distributions. Matching between marginals $g^{\left(i\right)}\left(x_{i}\right)$
and $g^{\left(a\right)}\left(x_{i}\right)$ results in 
\begin{align}
\frac{\mu_{i}}{\Sigma_{ii}} & =\text{atanh}\left\langle x_{i}\right\rangle _{g^{\left(a\right)}}\label{eq:DC}
\end{align}
$\forall i\in\partial a,a\in A$, giving $\sum_{a}\left|\partial a\right|$
equations. As we will see, \eqref{eq:DC} is chosen because it ensures
exactness on acyclic graphs. We call \emph{Density Consistency }(DC)
any scheme that enforces Eq. \eqref{eq:DC}. We propose to complement
Equation \eqref{eq:DC} with matching of first moments and Pearson
correlation coefficients corr$_{Q}\left(x,y\right)=(\left\langle xy\right\rangle _{Q}-\left\langle x\right\rangle _{Q}\left\langle y\right\rangle _{Q})(\left\langle x^{2}\right\rangle _{Q}-\left\langle x\right\rangle _{Q}^{2})^{-\frac{1}{2}}(\left\langle y^{2}\right\rangle _{Q}-\left\langle y\right\rangle _{Q}^{2})^{-\frac{1}{2}}$
(although other closures are possible, see \ref{subsec:Other-closure-equations})
\begin{equation}
\mu_{i}=\left\langle x_{i}\right\rangle _{g^{\left(a\right)}}\quad,\hspace{1em}\frac{\Sigma_{ij}}{\sqrt{\Sigma_{ii}\Sigma_{jj}}}=\rho\text{corr}_{g^{\left(a\right)}}\left(x_{i},x_{j}\right)\label{eq:DCT}
\end{equation}
for $i\neq j$ where $\rho\in\left[0,1\right]$ is an interpolating
parameter that is fixed to $1$ for the time being. Relations \eqref{eq:DC}-\eqref{eq:DCT}
give a system of $\sum_{a}\left|\partial a\right|\left(\left|\partial a\right|+3\right)/2$
equations and unknowns, that can be solved iteratively to provide
an approximation for the first $\left\langle x_{i}\right\rangle _{p}$
and second moments $\left\langle x_{i}x_{j}\right\rangle _{p}$ of
the original distribution. In a parallel update scheme (in which all
factors parameters are update simultaneously) the computational cost
of each iteration is $O\left(N^{3}\right)$, dominated by the calculation
of $\Sigma$. 

\begin{figure*}
\includegraphics[width=0.33\textwidth]{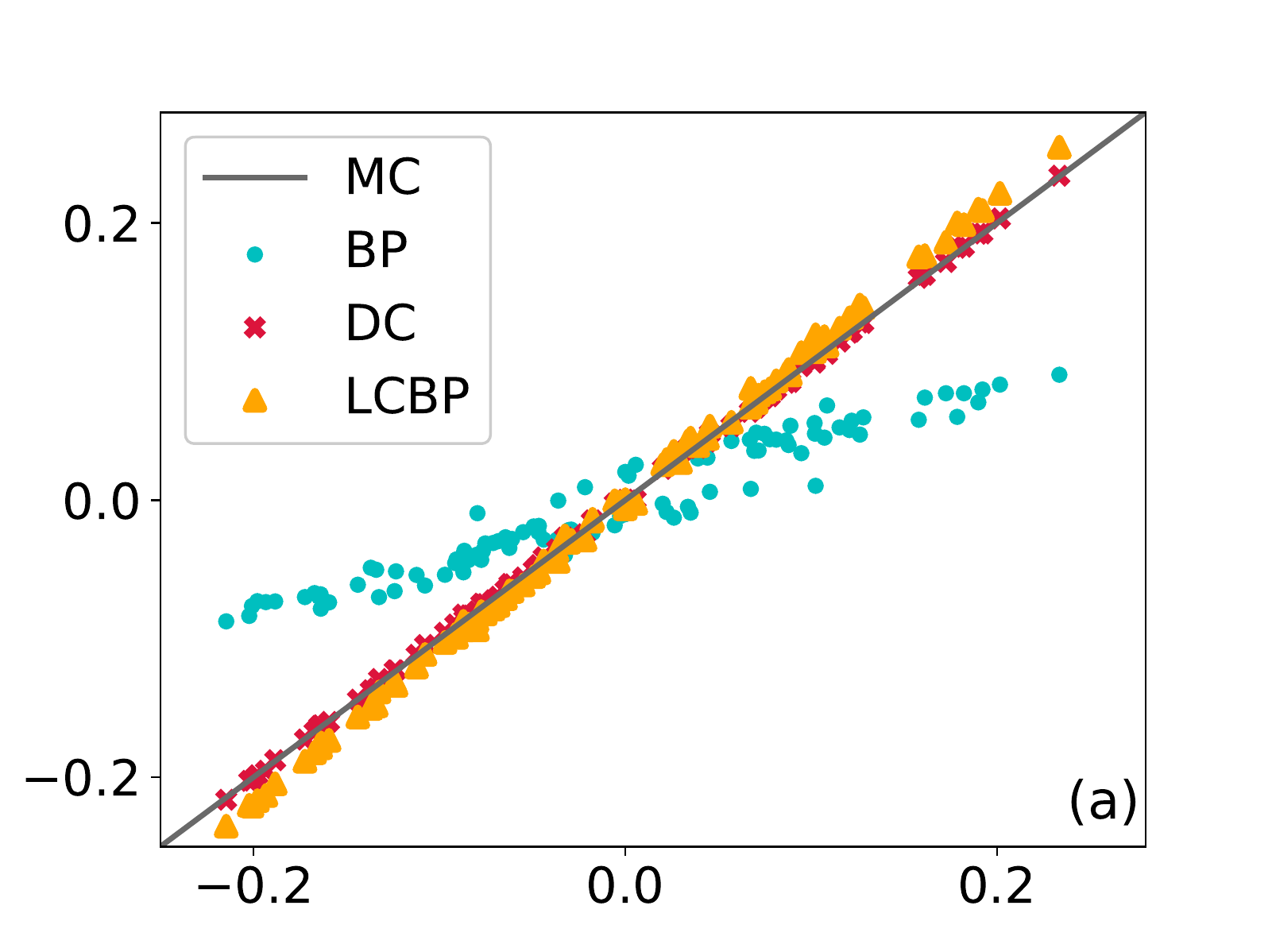}\includegraphics[width=0.33\textwidth]{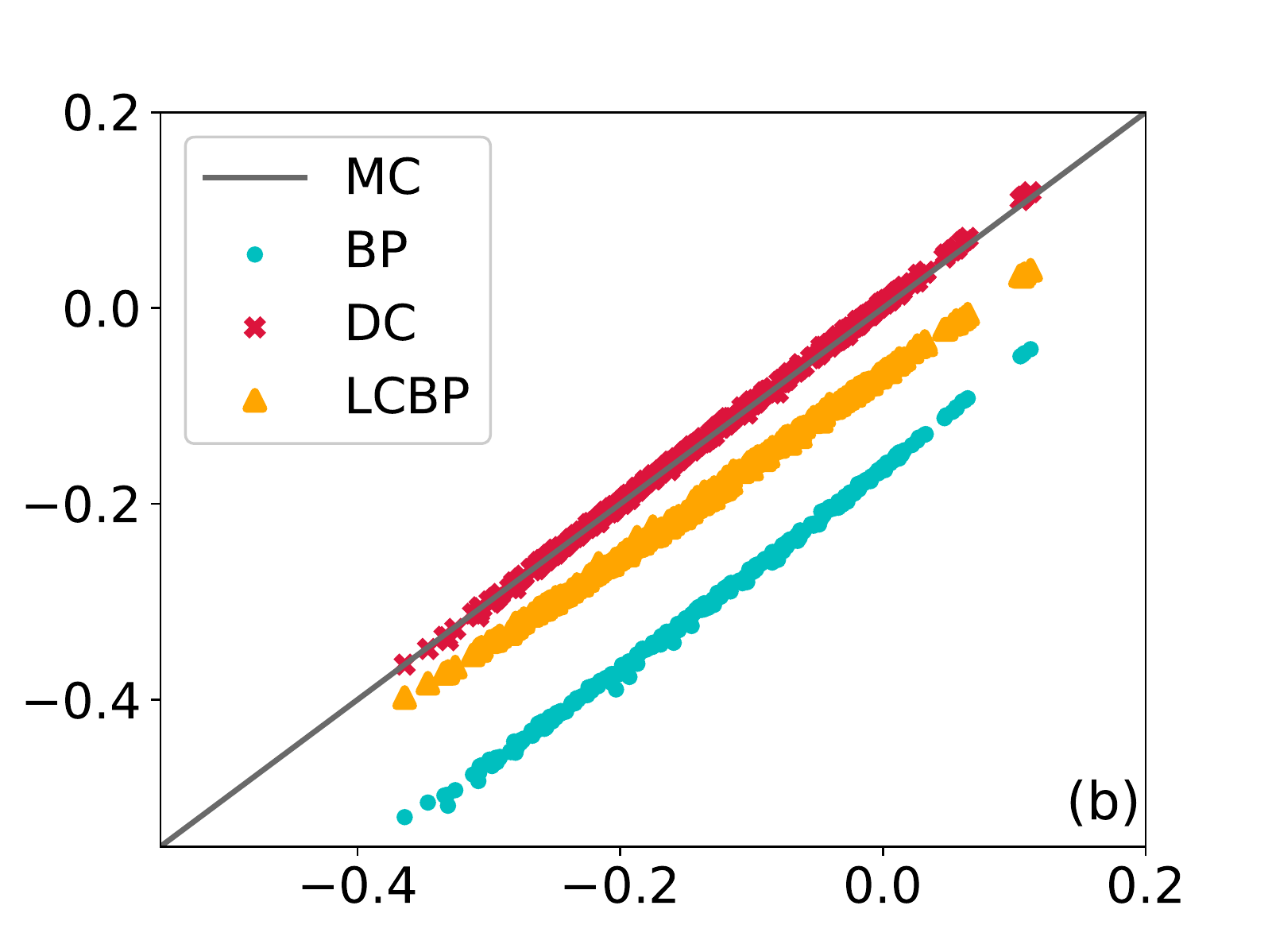}\includegraphics[width=0.33\textwidth]{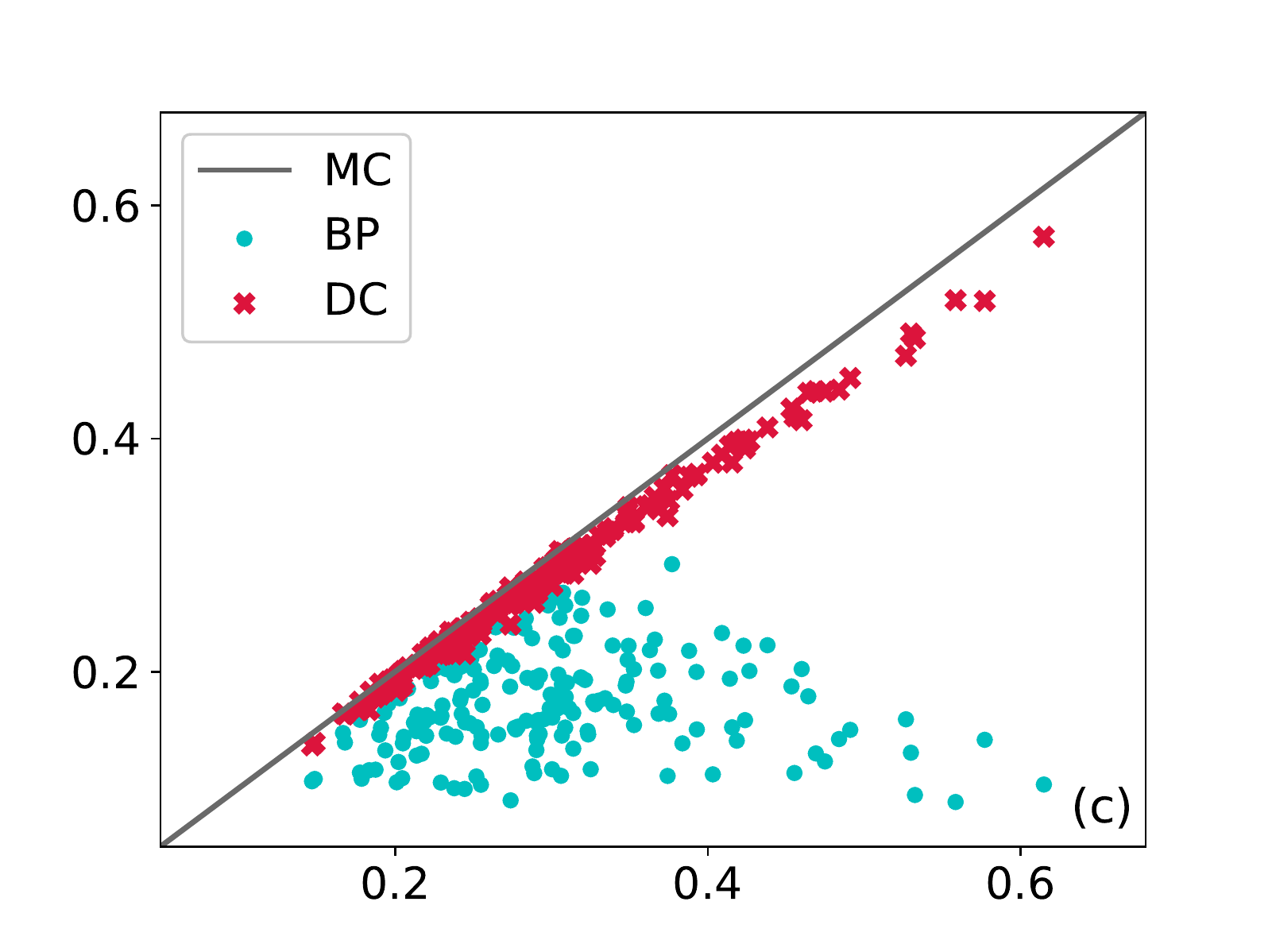}

\includegraphics[width=0.33\textwidth]{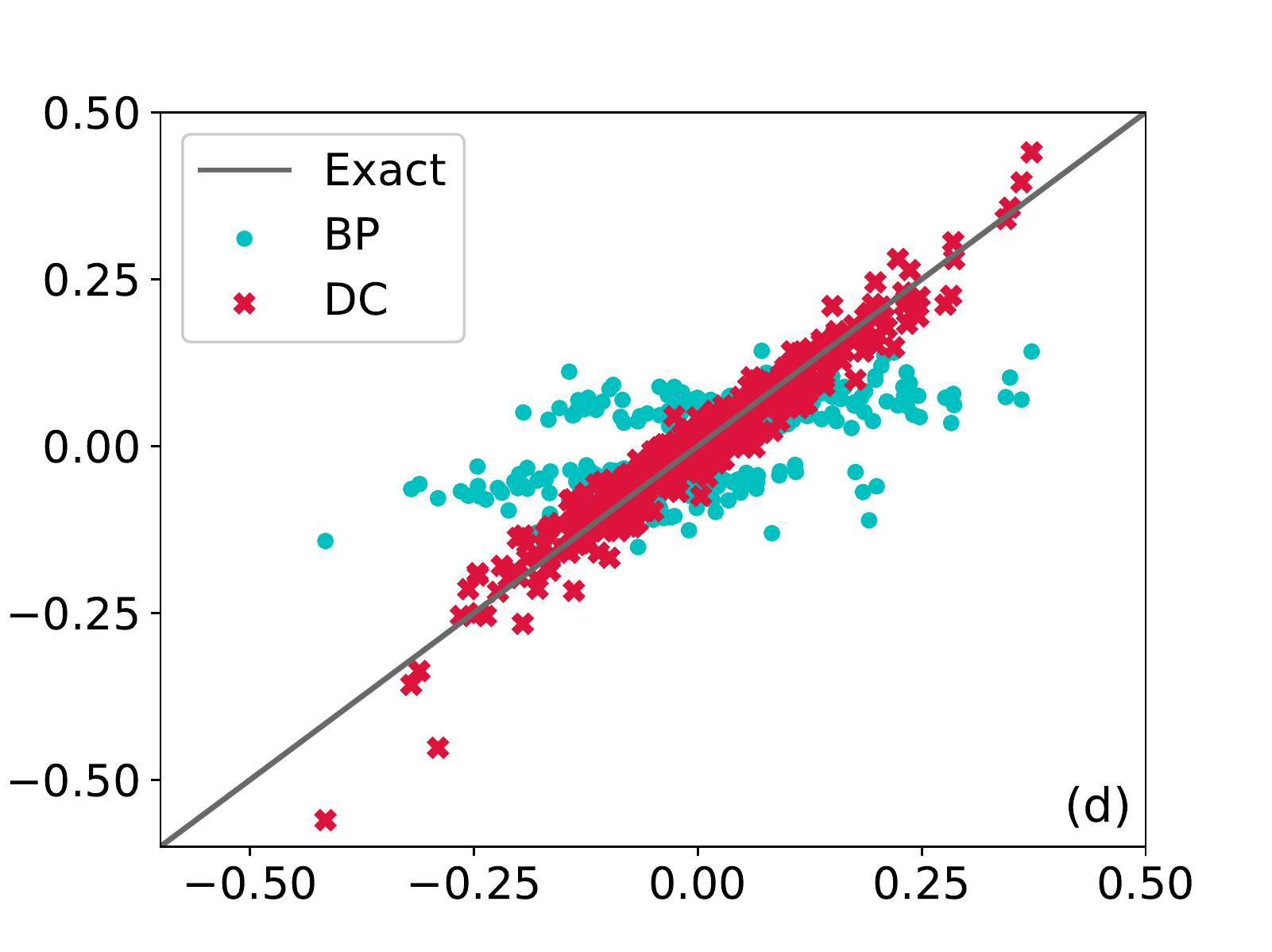}\includegraphics[width=0.33\textwidth]{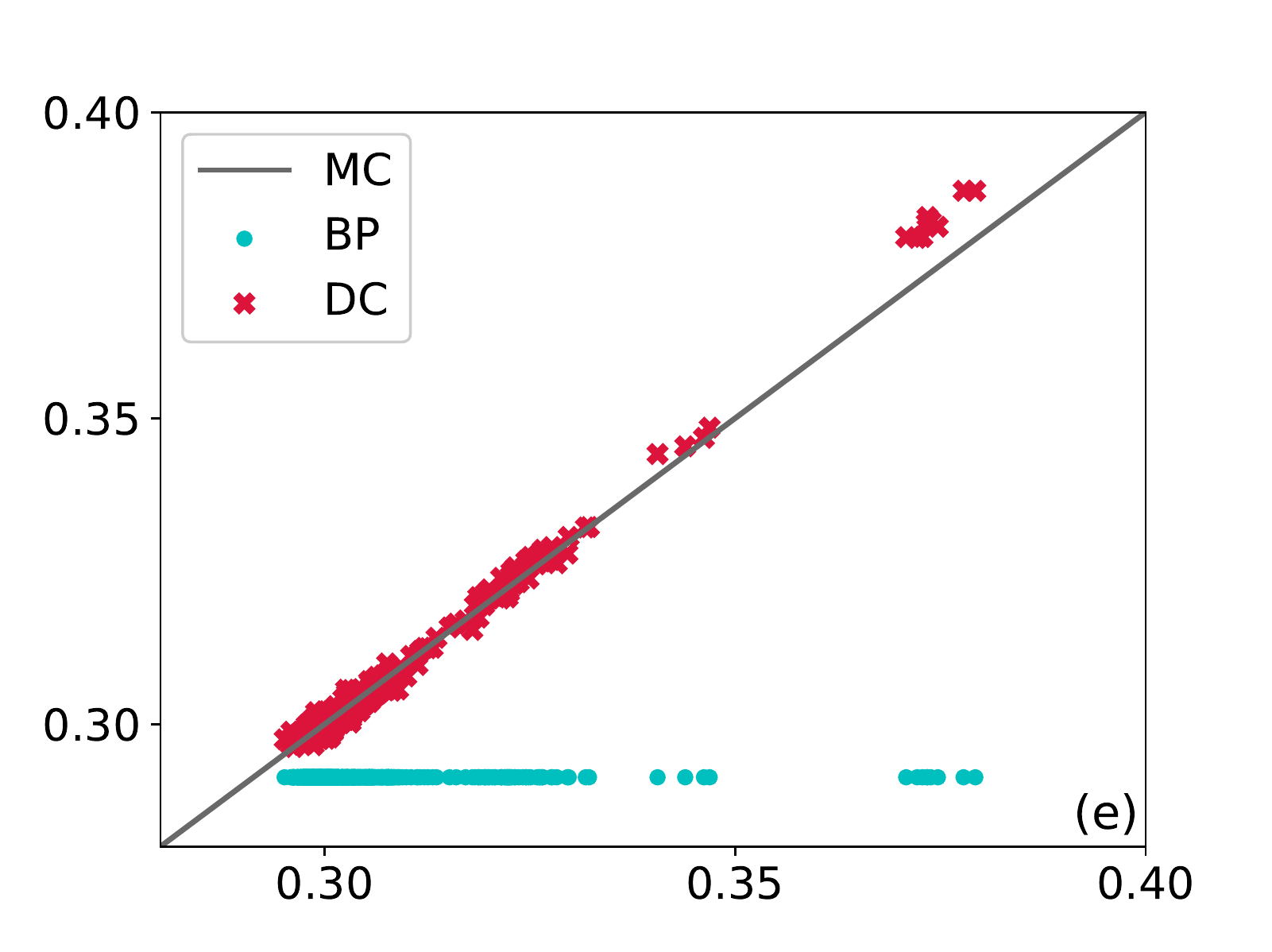}\includegraphics[width=0.33\textwidth]{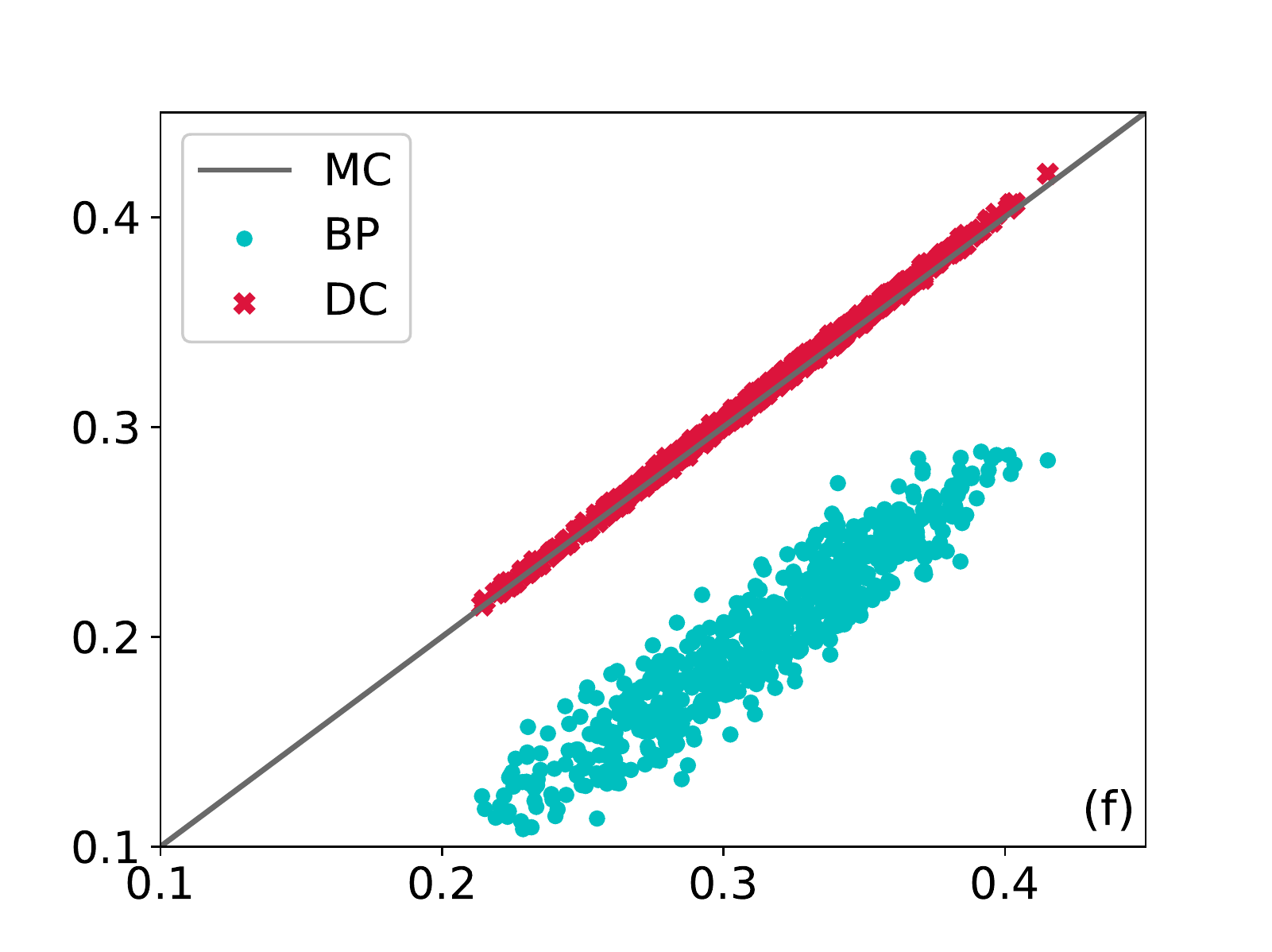}\caption{Comparison of DC, BP and LCBP on single-instances of disordered systems.
(a) Magnetizations of AntiFerromagnetic Ising Model on a triangular
lattice with $N=100,\left|E\right|=6N$, $J=-1$, $\beta=0.52$ and
random binary fields of $|h_{i}|=0.2$. (b) Magnetization of Ferromagnetic
Ising Model on a Random Regular (RR) Graph, $N=300$, degree $4$,
$\beta=0.35$, $J=1$ and random binary fields of $|h_{i}|=0.3$.
(c) Correlations of (heterogeneous) Ising Model on Barabasi-Albert
graph, $N=100$, $n_{0}=k=2$ without external fields (the solution
is found by using $\rho^{*}=0.95$ and it is divergent for $\rho>\rho^{*}$).
(d) Magnetizations of a random $4-$SAT instance at $\alpha=\frac{M}{N}=4$
at $\beta\to\infty$. (e) Correlations of (heterogeneous) Ferromagnetic
Ising Model on a Random Regular (RR) Graph, $N=300$, degree $4$
and $\beta=0.3$. (f) Correlations on a 3D hypercubic toroidal lattice
ferromagnetic (heterogeneous) Ising Model, $N=6^{3}$ and $\beta=0.21$
and no external fields. In heterogeneous ferromagnetic models, couplings
are drawn from a uniform distribution in $\left(0.5,1.5\right)$.
\label{fig:Heterogeneous}}
\end{figure*}
On acyclic factor graphs, the method converges in a finite number
of iterations and is exact, i.e. on a fixed point, $\left\langle x_{i}\right\rangle _{g}$
equals to the magnetization $\left\langle x_{i}\right\rangle _{p}$.
Therefore, as both the DC scheme and BP are exact on acyclic graphs,
their estimation of marginals must coincide. However, a deeper connection
can be pointed out. If a DC scheme applies zero covariances (e.g.
by setting $\rho=0$ in \eqref{eq:DCT}), on a DC fixed point on any
factor graph the quantities $m_{ai}=\tanh\left(\mu_{i}^{a}/\Sigma_{ii}^{a}\right)$
satisfy the Belief Propagation (BP) equations. Moreover, DC follows
dynamically a BP update. In particular, when equations converge, magnetizations
$m_{i}=\tanh\left(\mu_{i}/\Sigma_{ii}\right)$ are equal to the corresponding
belief magnetizations (Proof in \ref{subsec:RelationwithBP}).

Interestingly, the DC scheme can be thought of as a Gaussian pairwise
EP scheme with a modified consistency condition. The latter can be
obtained by keeping \eqref{eq:DCT} and replacing $\text{atanh}\left(x\right)$
in the RHS of \eqref{eq:DC} by the qualitatively similar $\frac{x}{1-x^{2}}$.
This of course renders the method inexact on acyclic graphs and turns
out to give generally a much worse approximation in many cases (See
\ref{subsec:Other-closure-equations}). In addition, as it also happens
with the EP method, Gaussian densities in factors $\psi_{a}$ can
be moved freely between factors (sharing the same variables) without
altering the approximation (details in \ref{subsec:Weight-gauge}).

\paragraph*{Numerical results}

We tested the method on the Ising model on many different scenarios
on heterogeneous systems, with a selection of results given in Figure
\ref{fig:Heterogeneous}. True values for magnetization and correlations
were computed approximately with long Monte-Carlo runs ($1\times10^{6}N-2\times10^{6}N$
Monte-Carlo Gibbs-sampling steps) for Ising models and with the exact
(exponential) trace for up to $N=28$ in the case of $k-$SAT. All
simulations have been performed with a damping parameter $\text{around }0.95$
to improve convergence. The DC method provides a substantial correction
to BP magnetizations and correlations in almost all cases ; it also
improves single-node marginal estimates w.r.t. Loop Corrected Belief
Propagation (LCBP) \citep{mooij2007loop} in several cases. LCBP simulations
were performed using the code provided in \citep{Mooij_libDAI_10}.
We underline that, despite the computational cost per iteration of
LCBP on bounded-degree graphs being $O\left(N^{2}\right)$, the prefactor
depends strongly on the degree distribution (with even exponential
scaling in some cases), also the number of iterations required to
converge is normally much larger than the one of DC. For instance,
for antiferromagnetic models (like the one shown in \ref{fig:Heterogeneous}(a)),
LCBP seems not to converge at smaller temperatures. 

\paragraph*{Homogeneous Ising model}

Consider a homogeneous ferromagnetic Ising Model with coupling constant
$J$ and external field $h^{ext}$ on on a $d$-dimensional lattice
with periodic (toroidal) boundary condition: because of the translational
invariance, all Gaussian factors $\phi_{a}$ are identical and the
covariance matrix admits an analytic diagonalization. Therefore it
is possible to estimate equilibrium observables through an analytical
DC scheme also in the thermodynamic limit. After some calculations
(see \ref{subsec:Simplified-DC-equations}), at a given temperature
the DC solution is found by solving the following system of 3 fixed
point equations $\sigma_{0}=\frac{m}{\text{atanh}m}$, $\sigma_{1}=\rho\frac{c-m^{2}}{1-m^{2}}\sigma_{0}$
and $y=m\left(\gamma_{0}+\gamma_{1}\right)$ in the Gaussian parameters
$y,\gamma_{0},\gamma_{1}$ where $m=\left\langle x_{i}\right\rangle _{g^{\left(a\right)}}$,
$c=\left\langle x_{i}x_{j}\right\rangle _{g^{\left(a\right)}}$ are
the moments computed under the ``tilted'' distribution $g^{\left(a\right)}$,
and $\sigma_{0},\sigma_{1},\gamma_{0},\gamma_{1}$ equal respectively
$\Sigma_{ii},\Sigma_{ij},\left(2d\right)^{-1}\left(\Sigma^{-1}\right)_{ii},\left(\Sigma^{-1}\right)_{ij}$
for $i,j$ two first lattice neighbors. Defining $R_{d}\left(r\right)=\frac{1}{2}\int_{0}^{\infty}dte^{-dt}\mathcal{I}_{0}^{d}\left(rt\right)$,
where $\mathcal{I}_{0}$ is the modified Bessel function of the first
kind of order $0$, and after some straightforward algebraic manipulations,
we finally obtain the following equations (here $h^{ext}=0$ for
simplicity) for variable $\beta,m,r=\gamma_{1}\gamma_{0}^{-1}$ 
\begin{align}
\beta= & \text{atanh}\left[\frac{1}{\rho}k_{d}\left(r\right)\left(1-m^{2}\right)+m^{2}\right]-g_{d}\left(r\right)\frac{\text{atanh}m}{m}-\text{atanh}\left[\tanh^{2}\left(f_{d}\left(r\right)\text{atanh}m\right)\right]\label{eq:DCIsing1}\\
m= & \tanh\left[f_{d}\left(r\right)\text{atanh}m+\text{atanh}\left(\tanh\left(\beta+g_{d}\left(r\right)\frac{\text{atanh}m}{m}\right)\tanh\left(f_{d}\left(r\right)\text{atanh}m\right)\right)\right]\label{eq:DCIsing2}
\end{align}
where $k_{d}\left(r\right)=\frac{1-2dR_{d}\left(r\right)}{r2dR_{d}\left(r\right)}$,
$g_{d}\left(r\right)=\frac{k_{d}\left(r\right)}{1-k_{d}\left(r\right)^{2}}+rR_{d}\left(r\right)$,
$f_{d}\left(r\right)=\frac{1}{1+k_{d}\left(r\right)}-\left(r+1\right)R_{d}\left(r\right)$.
Substituting \eqref{eq:DCIsing1} into \eqref{eq:DCIsing2} we get
a single equation for $m,r$, allowing for a parametric solution $m\left(r\right),\beta\left(r\right)$. 

The maximum value of $\beta$ for which a paramagnetic solution exists
can be analytically derived by substituting $m=0$ and taking $\sup_{-1<r\leq0}\beta\left(r\right)$
from \eqref{eq:DCIsing1}. For $d\geq3$ \footnote{$R_{2}\left(-1\right)$ is divergent and the maximum for $d=2$ is
located in $-1<r^{*}<0$. The solution for $d=2$ is thus qualitatively
different, see \ref{subsec:D=00003D2}.} the maximum is realized at $r=-1$, obtaining: 
\begin{align}
\beta_{p} & =\text{ath}\left(1-\frac{1}{x_{d}}\right)-\frac{x_{d}\left(x_{d}-1\right)}{2x_{d}-1}+\frac{x_{d}}{2d}\label{eq:betap}
\end{align}
where $x_{d}=2dR_{d}\left(-1\right)$. Values of $\beta_{p}$ for
various dimensions $d$ are reported in table I. The paramagnetic
solution is stable in the full range $0\leq\beta<\beta_{p}$ for $d\geq3$.

Expanding \eqref{eq:betap} in powers of $d^{-1}$ we get $\frac{1}{2d\beta_{p}}=1-\frac{1}{2}d^{-1}-\frac{1}{3}d^{-2}-\frac{13}{24}d^{-3}-\frac{979}{720}d^{-4}-\frac{2039}{480}d^{-5}+O\left(d^{-6}\right)$
which is exact up to the $d^{-4}$ order (the correct coefficient
of $d^{-5}$ is $-\frac{2009}{480}$) \citep{fisher_ising_1964}.
For comparison, nMF is exact up to the $d^{0}$ order, BP is exact
up to the $d^{-1}$ order, and \emph{Loop Corrected Bethe }(LCB, \citep{montanari_how_2005})
and Plaquette CVM (PCVM, \citep{dominguez_gauge-free_2017}) are exact
up to the $d^{-2}$ order.

\begin{table}
\begin{centering}
\begin{tabular}{cr@{\extracolsep{0pt}.}lr@{\extracolsep{0pt}.}lr@{\extracolsep{0pt}.}lr@{\extracolsep{0pt}.}lr@{\extracolsep{0pt}.}lr@{\extracolsep{0pt}.}l}
$d$ & \multicolumn{2}{c}{$\beta_{BP}$} & \multicolumn{2}{c}{$\beta_{PCVM}$} & \multicolumn{2}{c}{$\beta_{LCB}$} & \multicolumn{2}{c}{$\beta_{m}$} & \multicolumn{2}{c}{$\beta_{p}$} & \multicolumn{2}{c}{$\beta_{c}$}\tabularnewline
\hline 
2 & 0&34657 & \textbf{0}&\textbf{412258} & \multicolumn{2}{c}{-} & 0&388448 & 0&37693 & 0&440687\tabularnewline
3 & 0&20273 & 0&216932 & 0&238520 & 0&218908 & \textbf{0}&\textbf{222223} & 0&221654(6)\tabularnewline
4 & 0&14384 & 0&148033 & 0&151650 & 0&149835 & \textbf{0}&\textbf{149862} & 0&14966(3)\tabularnewline
5 & 0&11157 & 0&113362 & 0&114356 & \textbf{0}&\textbf{113946} & \textbf{0}&\textbf{113946} & 0&11388(3)\tabularnewline
6 & 0&09116 & 0&092088 & 0&092446 & \textbf{0}&\textbf{092304} & \textbf{0}&\textbf{092304} & (0&0922530)\tabularnewline
\end{tabular}
\par\end{centering}
\caption{\label{tab:Values}Critical values obtained with different approximation
schemes of the inverse temperature $\beta$ marking the onset of spontaneous
magnetization in the homogeneous Ising model on infinite $d$-dimensional
hypercubic lattices. The values of $\beta_{BP}$, $\beta_{PCVM}$
and $\beta_{LCB}$ respectively refer to the Bethe-Peierls, Plaquette
Cluster Variational Method \citep{dominguez_gauge-free_2017} and
Loop Corrected Bethe \citep{montanari_how_2005} approximations, while
$\beta_{p}$ and $\beta_{m}$ respectively correspond to the maximum
$\beta$ of the paramagnetic DC solution and the minimum $\beta$
of the magnetized DC solution, $\beta_{c}$ indicates the currently
best known approximation up to numerical accuracy (\citep{dominguez_gauge-free_2017}
for $d\protect\leq5$, \citep{fisher_ising_1964} for $d=6$). Results
in bold indicate the closest value to the last column.}
\end{table}
The minimum value of $\beta$ for which a magnetized solution exists
can be also computed by seeking a point with $\frac{d\beta}{dr}=0$
with the complication that $m$ is defined implicitly by \eqref{eq:DCIsing1}-\eqref{eq:DCIsing2}
(details in \ref{subsec:betam}). The resulting equation has a single
solution that has been numerically computed and shown in Table \eqref{tab:Values}
as $\beta_{m}$. It turns out to be smaller but always very close
to $\beta_{p}$ and coincident up to numerical precision for $d\geq5$.
Note that for inverse temperatures in the (albeit small) range $\beta_{m}<\beta<\beta_{p}$,
the DC approximation has both magnetized and a paramagnetic stable
solutions, suggesting a phase coexistence that should be absent in
the real system \citep{lebowitz_coexistence_1977}. 

\paragraph*{Discussion}

We proposed a general approximation scheme for distributions of discrete
variables that show interesting properties, including being exact
on acyclic factor graphs and providing a form of loop corrections
on graphs with cycles. 

In the same spirit as for PCVM and LCBP, DC approximation can be thought
of as a method to correct the cavity independence (or absence of cycles)
assumption in the Bethe-Peierls approximation. Whereas PCVM deals
only with local (short) cycles, it is true that LCBP and DC both attempt
to correct for arbitrarily long cycles in the interaction graph. However,
they do so through crucially different approaches. Loop Corrected
Belief Propagation (LCBP) works by computing several BP fixed points
(one for each cavity distribution in which one node and all the factors
connected to it are removed) and then imposing consistency over single-node
beliefs among them. Therefore, for each cavity distribution it computes
fixed points by still assuming a tree-factorization, i.e. by neglecting
correlations coming from other cycles in the graph. So it computes
a higher order approximation by relying on lower order ones (on a
modified/simplified) interaction graph. In this sense, it can be considered
as a first-order correction to BP and indeed it improves BP estimates
of single-node marginals, as shown in \ref{fig:Heterogeneous}. In
this perspective, DC can be considered as a new approximation in which
all 2-points cavity correlations are taken into account (of course,
in an approximate way, through a Gaussian distribution), in a single
self-consistent set of equations in which correlations arise simultaneously
from all cycles in the graph.

The method can be in part solved analytically on homogeneous systems
such as finite dimensional hypercubic lattices with periodic conditions.
Analytical predictions from the model show a number of interesting
features that are not shared by other mean-field approaches: the method
provides finite size corrections which are in close agreement with
numerical simulations; the paramagnetic solution exists only for $\beta<\beta_{p}$
(in PCVM and BP, the paramagnetic solution exists for all $\beta\geq0$,
although it stops being stable at a finite value of $\beta$); it
can capture some types of heterogeneity where the Bethe-Peierls approximation
can not (such as in RR graphs). Numerical simulations are in good
agreement on different models, including random-field Ising with various
topologies and random $k-$SAT. On lattices, the method could in principle
be rendered more accurate by taking into account small loops explicitly.
The DC scheme can be extended for models with $q-$states variables,
by replacing each of them with $q$ binary variables. Again, in this
setup it is possible to get a similar set of closure equations that
are exact on acyclic graphs and recover BP fixed points on any graph
when neglecting cavity correlations. This will be object of future
research. 
\begin{acknowledgments}
We warmly thank F. Ricci-Tersenghi, A. Lage-Castellanos, A. P. Muntoni,
I. Biazzo, A. Fendrik and L. Romanelli for interesting discussions.
AB and GC thank Universidad Nacional de General Sarmiento for hospitality.
The authors acknowledge funding from EU Horizon 2020 Marie Sklodowska-Curie
grant agreement No 734439 (INFERNET) and PRIN project 2015592CTH.
\end{acknowledgments}

\appendix

\section{General properties of DC scheme\label{sec:General-propertiesDC}}

\subsection{Relation with the Bethe Approximation (BP)\label{subsec:RelationwithBP}}

On acyclic graphs, both the DC scheme and BP are exact and thus they
must coincide on their computation of marginals. However, a deeper
connection can be pointed out. BP fixed point equations are
\begin{align}
m_{ai}\left(x_{i}\right)\propto & \sum_{\boldsymbol{x}_{a}}\psi_{a}\left(x_{a}\right)\prod_{j\in a\setminus i}m_{ja}\left(x_{j}\right)\label{eq:bp1}\\
m_{ia}\left(x_{i}\right)\propto & \prod_{b\in i\setminus a}m_{bi}\left(x_{i}\right)\label{eq:bp2}\\
m_{i}\left(x_{i}\right)\propto & \prod_{b\in i}m_{bi}\left(x_{i}\right)\label{eq:bp-marg}
\end{align}

\begin{thm}
If (H1) the DC scheme applies zero covariances or (H2) the factor
graph is acyclic, $m_{ia}\left(x_{i}\right)\propto g^{-a}\left(x_{i}\right)$
satisfies \eqref{eq:bp1}-\eqref{eq:bp2}. Moreover, the updates follow
dynamically BP updates. In particular, if equations converge, approximate
marginals $g\left(x_{i}\right)$ are proportional to belief magnetizations
\eqref{eq:bp-marg}. 
\end{thm}

\begin{proof}
In the either hypothesis (H1 or H2) , $g^{-a}\left(x_{a}\right)\propto\prod_{j\in\partial a}m_{ja}\left(x_{j}\right)$.
Define $m_{ai}\left(x_{i}\right)\propto\frac{g\left(x_{i}\right)}{m_{ia}\left(x_{i}\right)}$.
We obtain
\begin{align}
m_{ai}\left(x_{i}\right) & \propto\frac{1}{m_{ia}\left(x_{i}\right)}\int dx_{a\setminus i}g\left(x_{a}\right)\nonumber \\
 & \propto\int dx_{a\setminus i}\prod_{j\in\partial a\setminus i}m_{ja}\left(x_{j}\right)\phi_{a}\left(x_{a}\right)\label{eq:bpgauss}
\end{align}
Thanks to \eqref{eq:DC}, $g\left(x_{i}\right)\propto g^{\left(a\right)}\left(x_{i}\right)$
when $x_{i}\in\left\{ -1,1\right\} $ (and this is precisely the purpose
of \eqref{eq:DC}). In particular, for $x_{i}\in\left\{ -1,1\right\} $
we get also
\begin{align}
m_{ai}\left(x_{i}\right) & \propto\frac{1}{m_{ia}\left(x_{i}\right)}g^{\left(a\right)}\left(x_{i}\right)\nonumber \\
 & =\frac{1}{m_{ia}\left(x_{i}\right)}\sum_{x_{a\setminus i}}g^{-a}\left(x_{a}\right)\psi_{a}\left(x_{a}\right)\nonumber \\
 & =\sum_{x_{a\setminus i}}\prod_{j\in\partial a\setminus i}m_{ja}\left(x_{j}\right)\psi_{a}\left(x_{a}\right)\label{eq:dcbp1}
\end{align}
which is Eq. \eqref{eq:bp1}. Eq. \eqref{eq:bp2} is also verified
in either hypothesis:
\begin{enumerate}
\item Factorized case: if $\phi_{a}\left(x_{a}\right)=\prod_{i\in\partial a}\phi_{a}\left(x_{i}\right)$,
then clearly $m_{ai}\left(x_{i}\right)\propto\phi_{a}\left(x_{i}\right)$
and $m_{ia}\left(x_{i}\right)\propto\prod_{b\in\partial i\setminus a}m_{bi}\left(x_{i}\right)$
\item Acyclic case: if $T_{b}$ denotes the set of factors in the connected
component of $b$ once $i$ is removed, we get
\begin{align}
m_{ia}\left(x_{i}\right) & \propto g^{-a}\left(x_{i}\right)\nonumber \\
 & \propto\int dx_{-i}\prod_{b\in\partial i\setminus a}\phi_{b}\left(x_{b}\right)\prod_{c\in T_{b}\setminus b}\phi_{c}\left(x_{c}\right)\nonumber \\
 & \propto\prod_{b\in\partial i\setminus a}\int dx_{b\setminus i}\phi_{b}\left(x_{b}\right)\prod_{j\in\partial b\setminus i}g^{-b}\left(x_{j}\right)\nonumber \\
 & \propto\prod_{b\in\partial i\setminus a}\int dx_{b\setminus i}\phi_{b}\left(x_{b}\right)\prod_{j\in\partial b\setminus i}m_{bj}\left(x_{j}\right)\nonumber \\
 & \propto\prod_{b\in\partial i\setminus a}m_{bi}\left(x_{i}\right)\label{eq:dcbp2}
\end{align}
where the last line follows from \eqref{eq:bpgauss}.
\end{enumerate}
\end{proof}

\subsection{Relation with EP\label{subsec:Relation-with-EP}}

The DC scheme can be thought of a modified Gaussian EP scheme for
factors \footnote{With the apparent ambiguity of the appearance of $\delta^{k}$ terms
in $p\left(x\right)$, which is actually really not as problematic
as it may seem as the distribution is defined up to a normalization
factor; more precisely consider 
\[
p\left(x\right)=\lim_{\sigma\to0}\frac{1}{Z_{\sigma}}\prod_{i\sim j}\hat{\psi}_{\sigma,ij}\left(x_{i},x_{j}\right)
\]
for factors $\hat{\psi}_{\sigma,ij}\left(x_{i},x_{j}\right)=\psi_{ij}\left(x_{i},x_{j}\right)\left(\mathcal{N}\left(x_{i};1,\sigma\right)+\mathcal{N}\left(x_{i};-1,\sigma\right)\right)\left(\mathcal{N}\left(x_{j};1,\sigma\right)+\mathcal{N}\left(x_{j};-1,\sigma\right)\right)$.} 
\[
\hat{\psi}_{ij}\left(x_{i},x_{j}\right)=\psi_{ij}\left(x_{i},x_{j}\right)\left(\delta\left(x_{i}+1\right)+\delta\left(x_{i}-1\right)\right)\left(\delta\left(x_{i}+1\right)+\delta\left(x_{i}-1\right)\right)
\]
Classic EP equations in this context can be obtained by replacing
$\text{atanh}\left\langle x_{i}\right\rangle _{g^{\left(a\right)}}$
in the RHS of \eqref{eq:DC} by the qualitatively similar function
$\frac{\left\langle x_{i}\right\rangle _{g^{\left(a\right)}}}{1-\left\langle x_{i}\right\rangle _{g^{\left(a\right)}}^{2}}$,
but this of course invalidates Theorems 1-2 and turns out to give
a much worse approximation in general. 

\subsection{Weight gauge\label{subsec:Weight-gauge}}

One interesting property common to both DC and EP scheme concerns
the possibility to move freely gaussian densities in and out the exact
factors $\psi_{a}\left(x_{a}\right).$ Let $\rho_{a}\left(x_{a}\right)$
be Gaussian densities;
\begin{align*}
p\left(x\right) & \propto g\left(x\right)\prod_{a}\psi_{a}\left(x_{a}\right)\\
q\left(x\right) & \propto g\left(x\right)\prod_{a}\phi_{a}\left(x_{a}\right)
\end{align*}
and $q$ a Gaussian EP or DC approximation. We have

\begin{align}
p^{\left(a\right)}\left(x_{a}\right) & \propto\psi_{a}\left(x_{a}\right)\int dx_{-a}\frac{g\left(x\right)\prod_{b}\phi_{b}\left(x_{b}\right)}{\phi_{a}\left(x_{a}\right)}\label{eq:DC-1}\\
 & \propto\frac{\psi_{a}\left(x_{a}\right)}{\rho_{a}\left(x_{a}\right)}\int dx_{-a}\frac{\left[g\left(x\right)\prod_{b}\rho_{b}\left(x_{b}\right)\right]\prod_{b}\phi_{b}\left(x_{b}\right)/\rho_{b}\left(x_{b}\right)}{\phi_{a}\left(x_{a}\right)/\rho_{a}\left(x_{a}\right)}\\
q\left(x_{a}\right) & =\int dx_{-a}g\left(x\right)\prod_{b}\phi_{b}\left(x_{b}\right)\\
 & =\int dx_{-a}\left[g\left(x\right)\prod_{b}\rho_{b}\left(x_{b}\right)\right]\prod_{b}\phi_{b}\left(x_{b}\right)/\rho_{b}\left(x_{b}\right)
\end{align}

As DC and EP algorithms impose constraints between $p^{\left(a\right)}\left(x_{a}\right)$
and $q\left(x_{a}\right)$, any approximating family $\left\{ \phi_{a}\right\} $
for $\left(g,\left\{ \psi_{a}\right\} \right)$ leads to an equivalent
family $\left\{ \phi_{a}/\rho_{a}\right\} $ for $\left(g'=g\prod_{b}\rho_{b},\left\{ \psi'_{a}=\psi_{a}/\rho_{a}\right\} \right)$
for arbitrary factors $\rho_{a}$.

\subsection{Other closure equations\label{subsec:Other-closure-equations}}

Eq. \ref{eq:DC} is the only condition needed to make the approximation
scheme exact on tree-graphs. In principle one could complement it
with any other condition in order to obtain a well-determined system
of equations and unknowns in the factor parameters. In this work we
tried other complementary closure equations (including matching of
covariance matrix, constrained Kullback-Leiber Divergence minimization,
matching of off-diagonal covariances, in addition to \ref{eq:DC}).
However, we found out that \ref{eq:DCT} were experimentally performing
uniformly better on all the cases we analyzed.

\section{Homogeneous Ising Model\label{sec:Homogeneous-Ising-Model}}

In a homogeneous ferromagnetic Ising Model with hamiltonian $\mathcal{H}=-J\sum_{\langle i,j\rangle}x_{i}x_{j}-h^{ext}\sum_{i}x_{i}$
defined on a $d$-dimensional hypercubic lattice with periodic (toroidal)
boundary condition, all Gaussian factors $\phi_{a}$ are identical
and the DC equations (\ref{eq:DCT}\ref{eq:DC}) at a given inverse
temperature $\beta$ read: 
\begin{align}
\sigma_{0} & =\frac{m}{\text{atanh}m}\nonumber \\
\sigma_{1} & =\rho\frac{c-m^{2}}{1-m^{2}}\sigma_{0}\label{eq:startingDC}\\
y & =m\left(\gamma_{0}+\gamma_{1}\right)\nonumber 
\end{align}
The DC solution is found by solving the above system of 3 fixed-point
equations in the Gaussian parameters $y,\gamma_{0},\gamma_{1}$ where
$\sigma_{0},\sigma_{1},\gamma_{0},\gamma_{1}$ equal respectively
$\Sigma_{ii}$, $\Sigma_{ij}$, $\left(2d\right)^{-1}$$\left(\Sigma^{-1}\right)_{ii}$,
$\left(\Sigma^{-1}\right)_{ij}$ for $i,j$ two first lattice neighbors.
Here $m=\left\langle x_{i}\right\rangle _{g^{\left(a\right)}}$ and
$c=\left\langle x_{i}x_{j}\right\rangle _{g^{\left(a\right)}}$ are
the moments computed under the distribution $g^{\left(a\right)}$:
\begin{align*}
m & =\tanh\left[z+\text{atanh}\left(\tanh\Gamma\tanh z\right)\right]\\
c & =\tanh\left[\Gamma+\text{atanh}\left(\tanh^{2}z\right)\right]
\end{align*}
where 
\begin{align*}
z & =\frac{\beta h^{ext}}{2d}+\left(\frac{1}{\sigma_{0}+\sigma_{1}}\frac{1}{\gamma_{0}+\gamma_{1}}-1\right)y\\
\Gamma & =\beta J+\frac{\sigma_{1}}{\sigma_{0}^{2}-\sigma_{1}^{2}}+\gamma_{1}
\end{align*}
The matrix $\Sigma$ is the gaussian covariance matrix whose inverse
is parametrized as follows:
\[
\Sigma^{-1}=\mathcal{S}^{\left(d\right)}=2d\gamma_{0}\mathbb{I}_{L^{d}}+\gamma_{1}\mathcal{A}^{\left(d\right)}
\]
where $\mathcal{A}^{\left(d\right)}$ is the lattice adjacency matrix
in dimension $d$, whose diagonalization is discussed in the next
section.

\subsection{Diagonalization of $A^{\left(d\right)}$\label{subsec:Diagonalizationexact}}

The hypercubic lattice in $d$ dimensions can be regarded as the cartesian
product of linear-chain graphs, one for each dimension. The adiacency
matrix of the whole lattice can be thus expressed as function of the
adiacency matrices of the single linear chains, by means of the Kronecker
product (indicated by $\otimes$):

\[
\mathcal{A}^{\left(d\right)}=\mathcal{A}^{\left(1\right)}\otimes\mathcal{\mathbb{I}}_{L}\otimes...\mathcal{\otimes\mathbb{I}}_{L}+\mathbb{I}_{L}\otimes\mathcal{A}^{\left(1\right)}\otimes\mathbb{I}_{L}\otimes...\mathbb{I}_{L}+...+\mathbb{I}_{L}\otimes...\otimes\mathbb{I}_{L}\otimes\mathcal{A}^{\left(1\right)}
\]
where $\mathcal{A}^{\left(1\right)}$ is the adiacency matrix of a
(closed) linear chain of size $L$:
\[
\mathcal{A}^{\left(1\right)}=\begin{bmatrix}0 & 1 & 0 & \cdots & 0 & 0 & 1\\
1 & 0 & 1 & \cdots & 0 & 0 & 0\\
0 & 1 & 0 & \cdots & 0 & 0 & 0\\
\vdots & \vdots & \vdots & \ddots & \vdots & \vdots & \vdots\\
0 & 0 & 0 & \cdots & 0 & 1 & 0\\
0 & 0 & 0 & \cdots & 1 & 0 & 1\\
1 & 0 & 0 & \cdots & 0 & 1 & 0
\end{bmatrix}
\]
The above expression allows to compute the spectral decomposition
of $\mathcal{A}^{\left(d\right)}$ just by knowing the spectrum of
the adiacency matrix of the linear chain. The matrix $\mathcal{A}^{\left(1\right)}$
is a special kind of circulant matrix and therefore it can be diagonalized
exactly {[}\citealp{daviscirculant}{]}. Its eigenvalues and eigenvectors
are shown below:
\[
\lambda_{x}^{\left(1\right)}=2\textnormal{cos}\left(\frac{2\pi}{L}x\right)\qquad\qquad\nu_{x}^{\left(1\right)}=\frac{1}{\sqrt{L}}\left(1,w_{x},w_{x}^{2},\ldots,w_{x}^{L-1}\right)
\]

where $x\in\left\{ 0,...,L-1\right\} $ and $w_{x}=e^{i\frac{2\pi}{L}x}$.

The spectral decomposition of $\mathcal{A}^{(d)}$ reads:
\begin{align}
\lambda_{\left(x_{1},\ldots,x_{d}\right)}^{(d)} & =\sum_{j=1}^{d}\lambda_{x_{j}}^{\left(1\right)}=2\sum_{j=1}^{d}\textnormal{cos}\left(\frac{2\pi}{L}x_{j}\right)\label{eq:eigenvd}\\
\nu_{\left(x_{1},\ldots x_{d}\right)} & =\varotimes_{j=1}^{d}\nu_{x_{j}}^{\left(1\right)}
\end{align}

We recall now the expression of the eigenvalues of $S^{\left(d\right)}$:
\[
\lambda_{\left(x_{1},\ldots,x_{d}\right)}=2d\gamma_{0}+2\gamma_{1}\sum_{j=1}^{d}\textnormal{cos}\left(\frac{2\pi}{L}x_{j}\right)
\]
The inverse matrix elements $\Sigma_{ii},\Sigma_{ij}$ can be computed
in a straightfoward way. In particular, in the thermodynamic limit
$\left(L\to\infty\right)$ their expressions read:
\begin{align}
\sigma_{0} & =\frac{1}{\gamma_{0}}R\left(r\right)\label{eq:sigma0}\\
\sigma_{1} & =\frac{1}{\gamma_{0}r}\left[\frac{1}{2d}-R\left(r\right)\right]\label{eq:sigma1}
\end{align}

where $r=\frac{\gamma_{1}}{\gamma_{0}}$ and $R_{d}\left(r\right)=\frac{1}{2}\int_{0}^{\infty}dte^{-dt}\left[\mathcal{I}_{0}\left(rt\right)\right]^{d}$,
where $\mathcal{I}_{0}$ is the modified Bessel function of the first
kind of order $0$.

\subsection{Simplified DC equations\label{subsec:Simplified-DC-equations}}

It is possible to simplify the original system \eqref{eq:startingDC}
in order to get a fixed point equation for the magnetization $m$.
By eliminating the variable $y$ and setting $J=1$ we get
\begin{align}
z= & m\left(\frac{1}{\sigma_{0}+\sigma_{1}}-\left(\gamma_{0}+\gamma_{1}\right)\right)+\frac{\beta}{2d}h^{ext}\nonumber \\
= & m\gamma_{0}\left(\frac{1}{R_{d}\left(r\right)+\frac{1}{r}\left[\frac{1}{2d}-R_{d}\left(r\right)\right]}-r-1\right)+\frac{\beta}{2d}h^{ext}\label{eq:z}\\
\Gamma= & \beta+\frac{\sigma_{1}}{\sigma_{0}^{2}-\sigma_{1}^{2}}+r\gamma_{0}\nonumber \\
= & \beta+\gamma_{0}\left(\frac{\frac{1}{r}\left(\frac{1}{2d}-R_{d}\left(r\right)\right)}{R_{d}^{2}\left(r\right)-\frac{1}{r^{2}}\left[\frac{1}{2d}-R_{d}\left(r\right)\right]^{2}}+r\right)\label{eq:gamma}
\end{align}
Now, putting together Eq.\eqref{eq:startingDC} with Eq.\eqref{eq:z}-\eqref{eq:gamma}
and the definitions \eqref{eq:sigma0}-\eqref{eq:sigma1} we get the
following system:

\begin{align}
\beta= & \text{atanh}\left[\frac{1}{\rho}k_{d}\left(r\right)\left(1-m^{2}\right)+m^{2}\right]-g_{d}\left(r\right)\frac{\text{atanh}m}{m}-\text{atanh}\left[\tanh^{2}\left(f_{d}\left(r\right)\text{atanh}m+\frac{\beta}{2d}h^{ext}\right)\right]\label{eq:DCIsing1SI}\\
m= & \tanh\bigg[f_{d}\left(r\right)\text{atanh}m+\frac{\beta}{2d}h^{ext}+\label{eq:DCIsing2SI}\\
 & +\text{atanh}\left(\tanh\left(\beta+g_{d}\left(r\right)\frac{\text{atanh}m}{m}\right)\tanh\left(f_{d}\left(r\right)\text{atanh}m+\frac{\beta}{2d}h^{ext}\right)\right)\bigg]\nonumber 
\end{align}
where $k_{d}\left(r\right)=\frac{1-2dR_{d}\left(r\right)}{2drR_{d}\left(r\right)}$,
$g_{d}\left(r\right)=\frac{k_{d}\left(r\right)}{1-k_{d}\left(r\right)^{2}}+rR_{d}\left(r\right)$,
$f_{d}\left(r\right)=\frac{1}{1+k_{d}\left(r\right)}-\left(r+1\right)R_{d}\left(r\right)$.
Such equations can be solved at fixed $r$ in the variables $\beta,m$.
For $h=0$ the system reduces to a single fixed point equation for
$m=M\left(m\left(r\right),r\right)$ while $\beta$ is fixed by \eqref{eq:DCIsing1SI}.

\subsubsection*{Computation of $\beta_{p}$}

For the paramagnetic solution $m=0$ (with $h=0$) we get the following
equation for $\beta\left(r\right)$:
\begin{align*}
\beta= & \text{atanh}\left(\frac{1}{\rho r}\left[\frac{1}{2dR_{d}\left(r\right)}-1\right]\right)-g\left(r\right)
\end{align*}
For $d\ge3$, the maximum value at which a paramagnetic solution exists
corresponds to the point $r=-1$. Therefore, the value of the critical
point $\beta_{p}$ is computed by taking the $r\to-1$ limit of Eq.\eqref{eq:DCIsing1SI}:
\begin{align}
\beta_{p}= & \text{atanh}\left(1-\frac{1}{z}\right)-z\left(\frac{z-1}{2z-1}\right)+\frac{z}{2d}\label{eq:betap-1}
\end{align}
with $z=2dR_{d}\left(-1\right).$

\subsubsection*{Computation of $\beta_{m}$\label{subsec:betam}}

Eqs. \eqref{eq:DCIsing1SI}-\eqref{eq:DCIsing2SI} implicitly define
a function $m\left(r\right)$ such that $M\left(m\left(r\right),r\right)=m$,
and thus also $\beta\left(r\right)=\beta\left(m\left(r\right),r\right)$.
We seek to find the point $m^{*}=m\left(r^{*}\right)$ and $\beta_{m}=\beta\left(m^{*},r^{*}\right)$
such that $\frac{d\beta}{dr}\left(m\left(r^{*}\right),r^{*}\right)=0$.
Taking the total derivative of $\beta\left(m\left(r\right),r\right)$
we get the equation to be solved 
\[
0=\frac{d\beta}{dr}=\frac{\partial\beta}{\partial r}+\frac{\partial\beta}{\partial m}\frac{dm}{dr}
\]

To compute $\frac{dm}{dr}$ we use its implicit definition,
\begin{align*}
0 & =\frac{d}{dr}\left\{ M\left(m\left(r\right),r\right)-m\left(r\right)\right\} \\
 & =\left(\frac{\partial M}{\partial m}\left(m\left(r\right),r\right)-1\right)\frac{dm}{dr}+\frac{\partial M}{\partial r}\left(m\left(r\right),r\right)\\
\frac{dm}{dr} & =-\frac{\frac{\partial M}{\partial r}\left(m\left(r\right),r\right)}{\frac{\partial M}{\partial m}\left(m\left(r\right),r\right)-1}
\end{align*}
to get finally the $2\times2$ system in variables $m,r$:
\begin{align}
M\left(m,r\right)-m & =0\label{eq:betaDCm1}\\
\frac{\partial\beta}{\partial r}\left(m,r\right)\left(\frac{\partial M}{\partial m}\left(m,r\right)-1\right)-\frac{\partial M}{\partial r}\left(m,r\right)\frac{\partial\beta}{\partial m}\left(m,r\right) & =0\label{eq:betam}
\end{align}

\subsubsection*{Stability}

The stability of a fixed point $m^{*}=m\left(r^{*}\right)$ can be
analyzed by computing $\frac{dM}{dm}\bigg|_{m^{*}}$. In particular,
starting from the system \eqref{eq:DCIsing1SI}-\eqref{eq:DCIsing2SI}
where $r$ is implicitly defined as $r=R\left(\beta,m\right),$the
instability occurs when $\frac{dM}{dm}\bigg|_{m^{*}}=1$. Writing
the original system using the definition of $r$ we get $m=M\left(m,R\left(\beta,m\right)\right)$
and $\beta=B\left(m,R\left(\beta,m\right)\right)$. The equation we
want to solve is
\begin{align*}
1 & =\frac{dM}{dm}=\frac{\partial M}{\partial m}+\frac{\partial M}{\partial r}\frac{\partial R}{\partial m}
\end{align*}

To compute $\frac{\partial R}{\partial m}$ we use again its implicit
definition:
\begin{align*}
0 & =\frac{\partial B}{\partial m}+\frac{\partial B}{\partial r}\frac{\partial R}{\partial m}\\
\frac{\partial R}{\partial m} & =-\frac{\frac{\partial B}{\partial m}}{\frac{\partial B}{\partial r}}
\end{align*}

The final system to solve is
\begin{align}
M\left(m,r\right)-m & =0\label{eq:stab1}\\
\frac{\partial M}{\partial m}-\frac{\partial M}{\partial r}\frac{\frac{\partial B}{\partial m}}{\frac{\partial B}{\partial r}} & =1\label{eq:stab2}
\end{align}

For $d\geq3$, the solution becomes unstable exactly at the point
$\left(r_{m},\beta_{m}\right)$ computed through \ref{eq:betaDCm1}-\ref{eq:betam}.

\subsection{D=2\label{subsec:D=00003D2}}

On a 2-dimensional square lattice, the DC solution is qualitatively
different w.r.t. $d\geq3$ because the function $R\left(r\right)$
is logarithmically divergent for $r\to-1$. In such case the maximum
value at which the paramagnetic solution exists $\left(\beta_{p}=0.37693\right)$
corresponds to the point $r_{p}=-0.994843$. The ferromagnetic solution
turns out to be stable for $r_{m}<r<0$ with $r_{m}=-0.99405$, corresponding
to $\beta_{m}=0.388448$ (the point $\left(r_{m},\beta_{m}\right)$
is found as a solution of Eq. \ref{eq:stab1}-\ref{eq:stab2}). Therefore
there exists a temperature interval $\beta_{p}<\beta<\beta_{m}$ in
which no stable DC solution can be found.

For finite size lattices the DC solution can still be found numerically,
showing similar performances with respect to CVM on both ferromagnetic
and spin glass models (Fig. \ref{fig:2dsquare}). However, especially
on ferromagnetic systems DC solution is numerically unstable close
to the transition $\beta_{p}$. One way to reduce numerical instability
in such region is to decrease the interpolation parameter $\rho$,
typically fixed to 1 for DC. Neverthless, the meaning of the DC($\rho$)
approximation in this case is not clear. 

\begin{figure}
\includegraphics[width=0.35\textwidth]{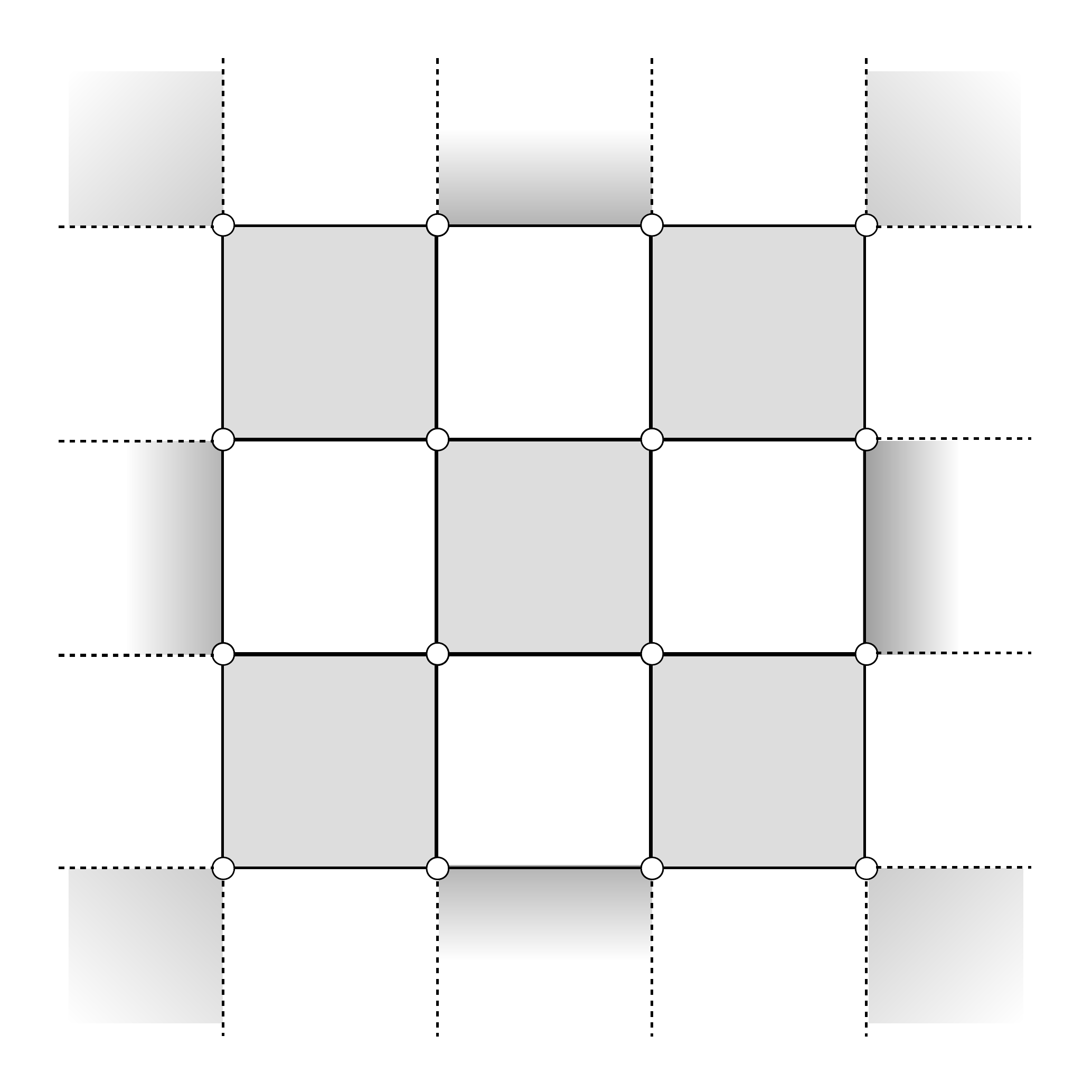}\includegraphics[width=0.5\textwidth]{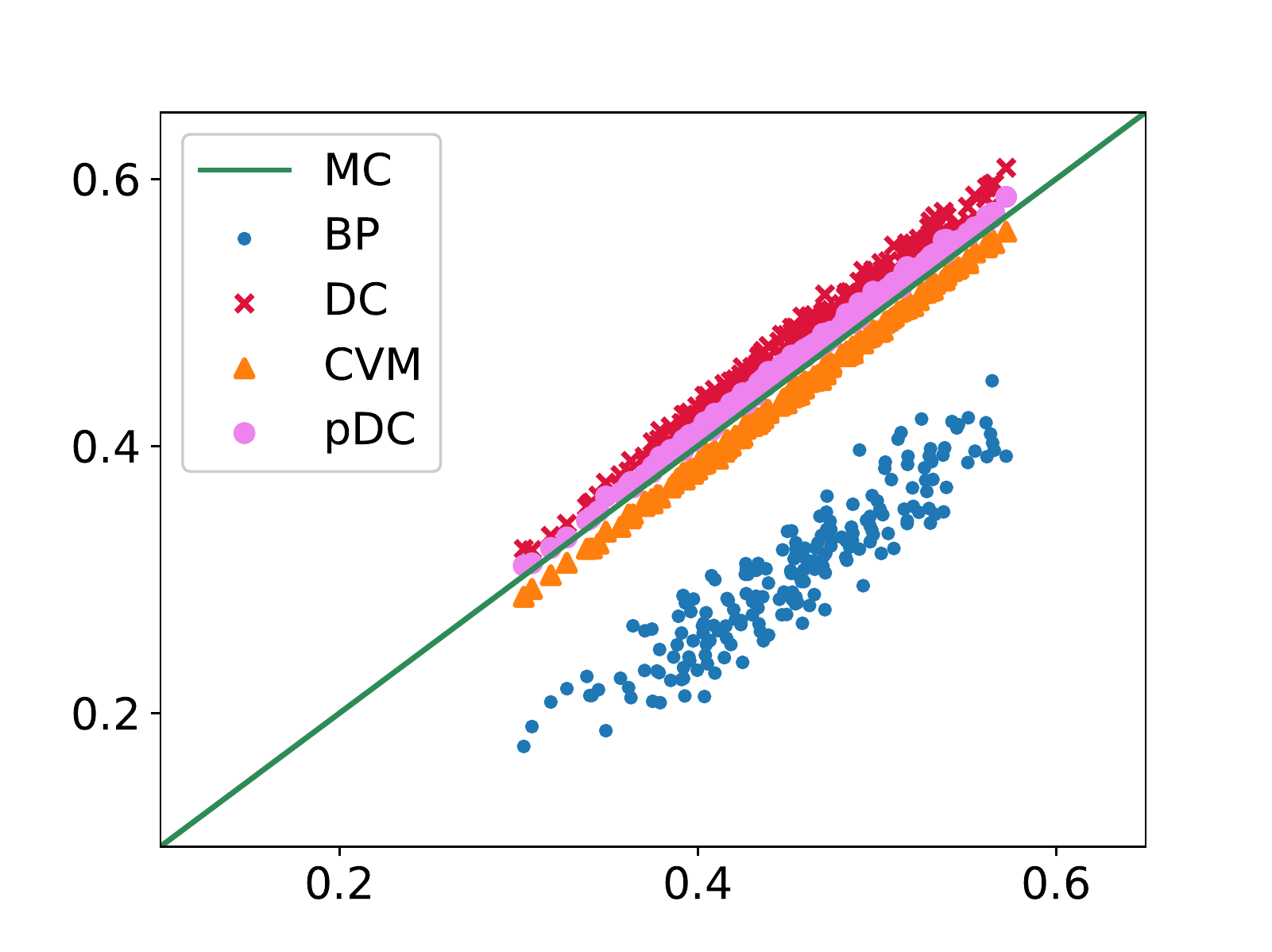}\caption{Left: distribution of non-overlapping plaquettes (in grey) on a 2-dimensional
square lattice. Right: correlations on 2-dimensional Ising Model on
square lattice of size $L=10$ at $\beta=0.36$, with $0$ external
field and couplings drawn from a uniform distribution in $\left(0.5,1.5\right)$.
Comparison of DC, pDC, BP and CVM\label{fig:2dsquare}}
\end{figure}
One possible way to improve the DC approximation is to take into account
small loops explicitly. In particular, we consider a gaussian family
of approximating distribution factorized over plaquettes of $2^{d}$spins
($d$ is the number of dimensions). Plaquettes are chosen in such
a way that there is no overlap between links in the gaussian distribution.
In this way, DC equation are exact on a plaquette tree with only site-overlaps.
Results are shown in \ref{fig:2dsquare}: plaquette-DC (pDC) is in
general slighty better than standard DC and comparable to CVM. 

\subsection{Finite size corrections}

In homogeneous models the gaussian covariance matrix can be diagonalized
analytically even for a finite size lattice (of size $L$). Therefore
we can compute finite size corrections to the DC solution at a fixed
$\beta$, as shown in the following plot:

\begin{figure*}
\includegraphics[width=0.6\textwidth]{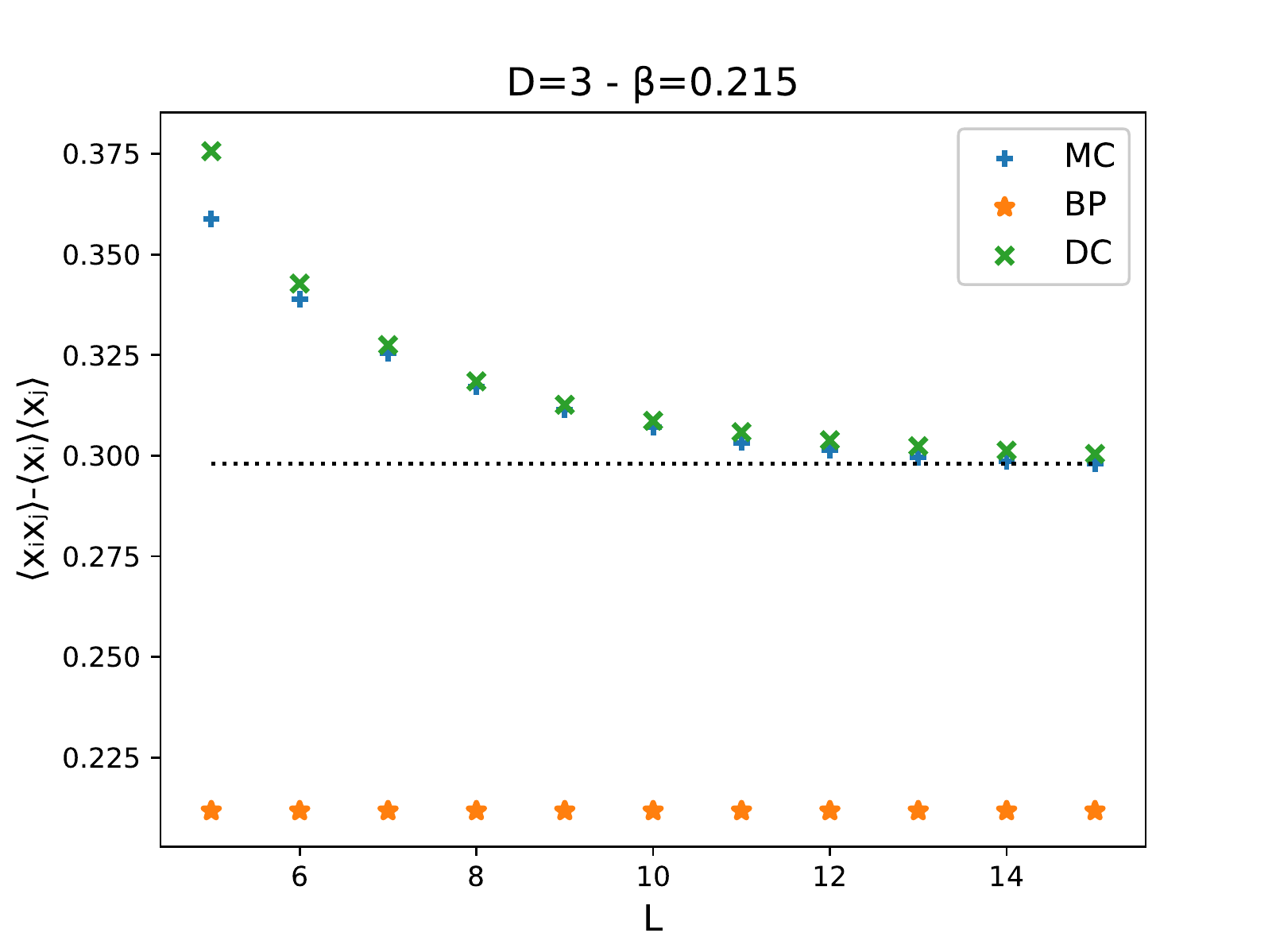}\caption{Finite size correction of equilibrium correlations at $\beta=0.215$
on a 3-dimensional cubic lattice of size $L\in\left\{ 5,...,15\right\} $
with $h=0,J=1$. Comparison of BP and DC solutions with Monte-Carlo
simulations.}
\end{figure*}
DC solution turns out to be in good agreement with MC results; on
the other hand, BP does not take into account at all finite size corrections
because of the local character of the approximation.

\subsection{Scaling of $\beta_{c}$ in the high dimensional limit}

Starting from the expression of the critical inverse temperature $\beta_{p}$
it is possible to compute the $1/d$ expansion in the high-dimensional
limit. We recall the expression of the critical temperature (\ref{eq:betap-1}):

\begin{align*}
\beta_{p}= & \text{atanh}\left(1-\frac{1}{z}\right)-z\left(\frac{z-1}{2z-1}\right)+\frac{z}{2d}
\end{align*}

where $z=2dR_{d}\left(-1\right)$.

Defining $x=1/d$ and expanding around $x=0$ we get:
\[
\frac{1}{2d\beta_{p}}=1-\frac{1}{2}d^{-1}-\frac{1}{3}d^{-2}-\frac{13}{24}d^{-3}-\frac{979}{720}d^{-4}-\frac{2039}{480}d^{-5}+O\left(d^{-6}\right).
\]

This expansion is exact up to the $d^{-4}$ order (the correct coefficient
of $d^{-5}$ is $-\frac{2009}{480}$) \citep{fisher_ising_1964}.
For comparison, Mean Field is exact up to the $d^{0}$ order, Bethe
is exact up to the $d^{-1}$ order, and Loop-Corrected Bethe and Plaquette-CVM
are exact up to the $d^{-2}$ order.

For the sake of completeness, we report the series expansion of $R_{d}\left(-1\right)$
around $x=0$:

\[
R_{d}\left(-1\right)=\frac{1}{2}d^{-1}+\frac{1}{4}d^{-2}+\frac{3}{8}d^{-3}+\frac{3}{4}d^{-4}+\frac{15}{8}d^{-5}+\frac{355}{64}d^{-6}+\frac{595}{32}d^{-7}+O\left(d^{-8}\right)
\]

\section{Multistates variables}

The method we presented is based on the possibility to fit the probability
values of a discrete binary distribution with the density values of
a univariate gaussian on the same support. When the model variables
take $q>2$ values there is no general way to fit single-node marginals
with a univariate Gaussian distribution. One possible solution is
to replace each $q$-state variable $x_{i}$ with a vector of $q$
(correlated) binary variables $\boldsymbol{s^{i}}$, where $s_{\alpha}^{i}\in\left\{ -1,1\right\} \forall\alpha=1,..,q$,
with the following constraint: 
\[
\sum_{\alpha=1}^{q}s_{\alpha}^{i}=2-q
\]
In this way, for each node $i$, only configurations of the type $\boldsymbol{s}^{i}=\left\{ 1,-1,...,-1\right\} $
(and its permutations) are allowed, in order to select just one of
the $q$ states for $x_{i}$. For each factor node $\emph{a}$, such
constraints can be implemented by adding a set of delta functions
in the original probability distribution, which is now a function
of the new binary variables $\boldsymbol{s}^{i}$. The correlations
induced by these constraints on the spin components of each $\boldsymbol{s^{i}}$
introduce short loops even when the original graph is a tree. Neverthless,
it is still possible to write a set of matching equation similar to
the 2-states case which is exact on trees.

\bibliographystyle{apsrev4-1}

\end{document}